\documentclass[journal]{IEEEtran}

\usepackage{epstopdf}
\usepackage[draft]{hyperref} 
\usepackage{setspace}
\usepackage{amsmath}
\usepackage{amssymb}
\usepackage{stfloats}
\usepackage{amsthm}
\usepackage{multirow}
\usepackage{amsfonts}
\usepackage{color}
\usepackage{graphicx}
\usepackage{subfigure}  
\usepackage[]{algorithmicx}
\usepackage{algpseudocode,algorithm}
\usepackage{cite}
\usepackage{mathtools}
\usepackage{stmaryrd}
\usepackage[mathscr]{euscript}
\usepackage{array}
\usepackage{mathrsfs}

\usepackage{soul}

\newtheorem{theorem}{Theorem}

\usepackage{authblk}

\author{   Hila Naaman$^{*}$, {\it Student Member, IEEE}, Satish~Mulleti$^{*}$, {\it Member, IEEE}, Yonina C. Eldar, {\it Fellow, IEEE}
\thanks{\scriptsize $^{*}$Authors contributed equally}
	\thanks{\scriptsize All the authors are with the Faculty of Math and Computer Science, Weizmann Institute of Science, Israel. Email: hila.naaman@weizmann.ac.il, satish.mulleti@gmail.com, yonina.eldar@weizmann.ac.il}
	\thanks{\scriptsize This project has received funding from the Benoziyo Endowment Fund for the Advancement of Science, Estate of Olga Klein -– Astrachanthe; Israeli Council for Higher Education (CHE) via the Weizmann Data Science Research Center and by a research grant from Madame Olga Klein – Astrachan; European Union’s Horizon 2020 research and innovation program under grant No. 646804-ERC-COG-BNYQ; and the Israel Science Foundation under grant no. 0100101. }

}

\begin{document}

\title{FRI-TEM: Time Encoding Sampling of Finite-Rate-of-Innovation Signals}
	\markboth{Submitted to the IEEE Transactions of Signal Processing}%
	{Shell \MakeLowercase{\textit{et al.}}: Bare Demo of IEEEtran.cls for Journals}
\maketitle

\begin{abstract}
Classical sampling is based on acquiring signal amplitudes at specific points in time, with the minimal sampling rate dictated by the degrees of freedom in the signal. The samplers in this framework are controlled by a global clock that operates at a rate greater than or equal to the minimal sampling rate. At high sampling rates, clocks are power-consuming and prone to electromagnetic interference.
An integrate-and-fire time encoding machine (IF-TEM) is an alternative power-efficient sampling mechanism which does not require a global clock. Here, the samples are irregularly spaced threshold-based samples. In this paper, we investigate the problem of sampling finite-rate-of-innovation (FRI) signals using an IF-TEM. We provide theoretical recovery guarantees for an FRI signal with arbitrary pulse shape and without any constraint on the minimum separation between the pulses. In particular, we show how to design a sampling kernel, IF-TEM, and recovery method such that the FRI signals are perfectly reconstructed. We then propose a modification to the sampling kernel to improve noise robustness. Our results enable designing low-cost and energy-efficient analog-to-digital converters for FRI signals.
\end{abstract}

%
\IEEEpeerreviewmaketitle

\section{Introduction}
Sampling is a process which enables discrete representation of continuous-time signals allowing for efficient processing of analog signals using digital signal processors \cite{eldar2015sampling,unser2000sampling}.
A commonly used discrete representation is to measure uniform, instantaneous samples of an analog signal, as in Shannon-Nyquist sampling theory \cite{nyquist1928certain}, and more general shift-invariant sampling \cite{eldar2009compressed,christensen2004oblique}. The amplitude samples are measured via a sample and hold circuit that is controlled by a global clock that operates at a speed greater or equal to the minimum sampling rate. At high sampling rates, clocks are power consuming and subject to electromagnetic interference \cite{asynchronous_adc}.   


Time encoding machines (TEM) provide an alternative digital representation of analog signals \cite{lazar2004perfect,adam2020sampling,rastogi2009low,alexandru2019reconstructing,hilton2021time,rudresh2020time,adam2020encoding,martinez2019delta}, which is asynchronous, that is, no global clock is required unlike conventional analog to digital converters. This leads to lower power consumption and reduced electromagnetic interference \cite{asynchronous_adc}. In this sampling scheme, an analog signal is represented by a set of time instants at which the input signal or its function crosses a certain threshold. The number of time instants per unit time or \emph{firing rate} is proportional to the local frequency content of the input signal. For example, there are more firings in a region where the signal amplitude varies rapidly compared to regions with relatively slow variations.

A popular approach for time encoding is an integrate and fire time encoding machine (IF-TEM), which is a brain-inspired sampling paradigm. It leads to simple and energy-efficient devices, such as analog-to-digital converters \cite{rastogi2009low,hilton2021time}, neuromorphic computers \cite{neuromorphic_computer}, event-based vision sensors \cite{camera1, camera2}, and more. In an IF-TEM, a bias is added to the analog input signal to make the signal positive. The resulting signal is then scaled and integrated, and the integral value is compared to a threshold. Each time the threshold is reached, time points or firing instants are recorded, which encode the information of the analog signal \cite{lazar2004perfect,lazar2004time}. A natural question is whether the analog signal can be perfectly reconstructed from the time-encodings.  

In recent years time-based sampling theory has witnessed growing interest with several authors proving the capabilities of TEM to sample and reconstruct bandlimited signals \cite{lazar2004time,lazar2003time,feichtinger2012approximate,adam2020sampling,adam2019multi,adam2019multi,adam2020encoding}. Interestingly, the minimum firing rate of a TEM for perfect recovery of a bandlimited signal is equal to the Nyquist rate of the signal. Since the firing rate has to increase with bandwidth, we look beyond the bandlimited structure of the signal so that sampling can be performed at a sub-Nyquist rate.

	\begin{figure}[t!]
		\centering
		\includegraphics[width= 2.5	in]{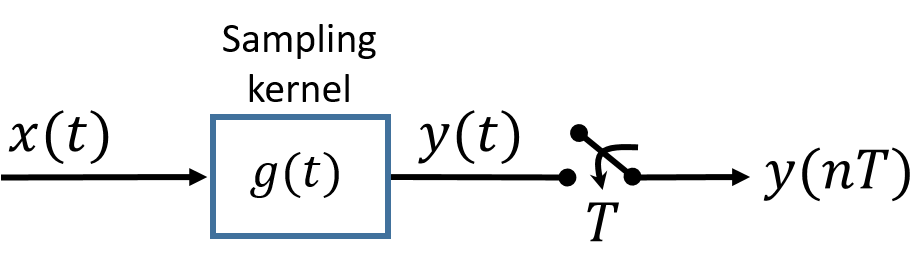}
		\caption{A kernel-based FRI sampling framework: An FRI signal $x(t)$ is first filtered by a sampling kernel $g(t)$ and then instantaneous uniform samples are measured at a sub-Nyquist rate. Parameters of the FRI signal are estimated from the sub-Nyquist samples.}
		\label{fig:FRIgeneric}
	\end{figure}

In the sub-Nyquist regime, finite-rate-of-innovation (FRI) signals are widely studied \cite{vetterli2002sampling, tur2011innovation, eldar2015sampling}. These signals have fewer degrees of freedom than the signal’s Nyquist rate samples, which enables sub-Nyquist sampling \cite{urigiien2012sampling}. For instance, consider an FRI signal consisting of a sum of $L$ amplitude-scaled and time-delayed copies of a known pulse. Since the pulse is known, the signal is completely specified by $L$ amplitudes and $L$ time-delays, which amounts to $2L$ degrees of freedom. It has been shown that $2L$ measurements of the signal uniquely determine the amplitude and the time-delays. This is equivalent to saying that these signals have a finite rate of innovation which equals $2L$.
A typical FRI sampling scheme is shown in Fig.~\ref{fig:FRIgeneric} where the signal is first filtered by a sampling kernel to remove redundancy in the signal, and then instantaneous samples are measured at a sub-Nyquist rate \cite{tur2011innovation}. Given the advantage of TEM over conventional sampling, we are interested in studying the applicability of TEM-based sampling to FRI signals. 

\begin{figure}
    \centering
    \includegraphics[width= 3 in]{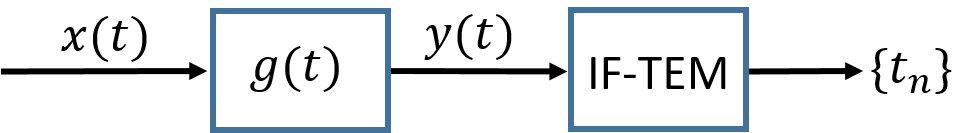}
    \caption{Sampling setup: A continuous-time signal $x(t)$ is filtered through a sampling kernel $g(t)$ and then sampled by a TEM which generates time-instants $\{t_n\}$.}
    \label{fig:TEMfri}
\end{figure}

We consider the TEM-based sampling scheme for FRI signals shown in Fig.~\ref{fig:TEMfri}, where the signal is modeled as a sum of shifted and scaled pulses with a known pulse shape.
Our goal is to develop conditions on the sampling kernel and IF-TEM parameters so that perfect recovery is guaranteed.

Alexandru and Dragotti \cite{alexandru2019reconstructing} 
consider a sequential reconstruction method for certain FRI signals. They show that by using either a compactly supported polynomial generating kernel or an exponential generating kernel, the time delays and amplitudes of each pulse can be perfectly recovered from four firing instants. Thus, to reconstruct a stream of $L$ pulses with $2L$ degrees of freedom, $4L$ firing instants are needed. The sequential nature of the reconstruction imposes a restriction on the minimum separation between any two consecutive pulses, such that any two successive pulses must be separated by the support of the sampling kernel. In addition, the threshold of the IF-TEM must be small enough to achieve a sufficient firing rate. To address these issues, Hilton et al. \cite{hilton2021time} consider IF-TEM sampling by using the derivative of a hyperbolic secant as a sampling kernel. They showed that a stream of $L$ Dirac impulses or piecewise constant functions with $L$ discontinuities are perfectly reconstructed from $3L+1$ firing instants without any minimum separation conditions. However, the sampling mechanism can not be extended to FRI signals with arbitrary known pulse shapes.  

Rudresh et al. \cite{rudresh2020time} show that by using sampling kernels that have a frequency-domain alias cancellation properties (see \cite{tur2011innovation} for details), FRI signals with $2L$ degrees of freedom can be recovered from $2L+2$ IF-TEM firing instants. The reconstruction algorithm does not require any minimum separation conditions but assumes that certain invertibility conditions are guaranteed. Through simulations, the authors show that the invertibility conditions are satisfied for a large number of experiments; however, theoretical guarantees are not given. In addition, their sampling results do not deal with the noisy scenario.     

Our contribution is twofold. First, we derive theoretical guarantees for perfect reconstruction of FRI signals with arbitrary but known pulse shapes. Second, we design the sampling kernel and IF-TEM sampler with improved noise robustness compared to existing approaches. 
Since FRI signals with $2L$ degrees of freedom can be perfectly reconstructed from $2L$ consecutive Fourier series coefficients (FSCs) \cite{eldar2015sampling, tur2011innovation}, we consider a frequency-domain FRI signal reconstruction approach. In particular, we choose sampling kernels with the alias-cancellation condition to annihilate the undesirable FSCs \cite{tur2011innovation}. The filtered signal, with fewer FSCs, is applied to an IF-TEM, and firing instants are measured. The measurements derived from the firing instants are linearly related to the FSCs of the filtered FRI signal \cite{rudresh2020time}. We show that $2L+2$ firing instants are sufficient for the relation to be invertible; that is, $2L$ FSCs are uniquely determined from a minimum of $2L+2$ IF-TEM measurements. Furthermore, we establish conditions on the IF-TEM parameters that ensure that the minimum firing rate is achieved. To summarize, we show that by using a sampling kernel with frequency-domain alias-cancellation properties and an IF-TEM sampler with a minimum firing rate of $2L+2$ per time unit, an FRI signal with $2L$ degrees of freedom is uniquely recovered.

While our first sampling approach leads to perfect reconstruction of the signal in the absence of noise, the reconstruction method can be highly sensitive to noise. To address this issue, we propose a modified sampling and reconstruction mechanism. 
In particular, we show that the zeroth Fourier coefficient of the filtered signal results in unstable inverse while computing the FSCs from the time instants in the presence of noise. To improve noise robustness while estimating FSCs from IF-TEM measurements, we modify the sampling kernel by removing the zero-frequency component.
For this modified method, we show that $4L+2$ time instants are sufficient for perfect recovery when the time-delays of the FRI signal are off-grid, whereas $2L+2$ firings are sufficient when the delays are on grid. This latter approach requires twice the number of firings for perfect recovery with off-grid time delays, but is more robust to noise. Through simulations, we show that for the same number of firing rates (beyond $2L+2$ firings), the mean squared error in the estimation of the time delays in the second approach is 2-6 dB lower compared to the first.

This paper is organized as follows. In Section II, we first review IF-TEM (Section \ref{subsec:TEM}), followed by a problem formulation in Section \ref{subsec:problem}. In Section III, we present our first results where the sampling kernel includes a zero-frequency component. A noise-robust sampling and reconstruction method together with simulations are presented in Section IV. While we consider periodic FRI signal model in the previous sections, in Section \ref{sec:Nonperiodic}, we discuss recovery guarantees for non-periodic FRI signal. Our concluding remarks presented in Section \ref{sec:Conclusions}.

	\begin{figure}[t!]
	\centering
	\includegraphics[width = 3 in]{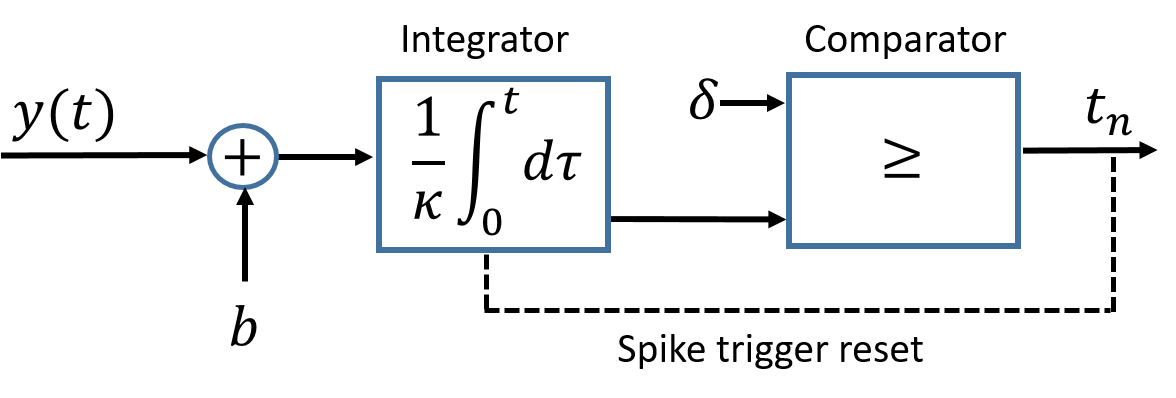}
	\caption{Time encoding machine with spike trigger reset. The input is biased by $b$, scaled by $\kappa$, and integrated. A time instant is recorded when the threshold $\delta$ is reached, after which the value of the integrator resets.}
	\label{fig:Vetterli_TEM}
\end{figure}
\section{Problem Formulation and Preliminaries}
We begin by presenting some known results on IF-TEM followed by the problem formulation.

\subsection{Time Encoding Machine}
\label{subsec:TEM}
We consider an IF-TEM whose operating principle is the same as in \cite{rudresh2020time} (see Fig. \ref{fig:Vetterli_TEM}). The input to the IF-TEM is a bounded signal $y(t)$, and the output is a series of firing or time instants. An IF-TEM is parametrized by positive real numbers $b$, $\kappa$, and $\delta$. This mechanism work as follows: A bias $b$ is added to a $c$-bounded signal $y(t)$ where $|y(t)|\leq c<b<\infty$, and the sum is integrated and scaled by $\kappa$. When the resulting signal reaches the threshold $\delta$, the time instant $t_n$ is recorded, and the integrator is reset. The process is repeated to record subsequent time instants, i.e., if a time instant $t_n$ was recorded, the next time instant $t_{n+1}$ is measured such that 
\begin{align}
   \frac{1}{\kappa}\int_{t_n}^{t_{n+1}} (y(s)+b)\, ds = \delta.
   \label{eq:y_t_relation_st}
\end{align}
The time encodings $\{ t_n, n\in\mathbb{Z} \}$ form a discrete representation of the analog signal $y(t)$ and the objective is to reconstruct $y(t)$ from them. Typically, the reconstruction is performed by using an alternative set of discrete representations $\{y_n, n \in \mathbb{Z}\}$ defined as
\begin{equation}
y_n \triangleq \int_{t_n}^{t_{n+1}}y(s)\, ds =  -b(t_{n+1}-t_n)+\kappa\delta.
\label{eq:trigger0}
\end{equation}
The measurements $\{y_n, n \in \mathbb{Z}\}$ are derived from the time encodings $\{ t_n, n\in\mathbb{Z} \}$ and IF-TEM parameters $\{b, \kappa, \delta\}$ \cite{lazar2004time, rudresh2020time}. 

Although  reconstruction methods vary for different classes of signals, for perfect recovery of any signal, the firing rate is required to satisfy a lower bound that depends on the degrees of freedom of the signal. The firing rate of an IF-TEM is bounded both from above and below, where the bounds are a function of the IF-TEM parameters and an upper-bound on the signal amplitude.  
Using \eqref{eq:trigger0} and the fact that $|y(t)|\leq c$, it can be shown that for any two consecutive time instants \cite{lazar2003time,lazar2004time}:
\begin{equation}
 \frac{\kappa\delta}{b+c} \leq t_{n+1}-t_n \leq \frac{\kappa\delta}{b-c}.
   \label{eq:consecutive_time}
\end{equation}
The inequalities in \eqref{eq:consecutive_time} imply that in any arbitrary, non-zero, observation interval $T_{\text{obs}}$, the maximum and minimum number of firings are
\begin{align}
    T_{\text{obs}} \frac{b+c}{\kappa \delta} \quad \text{and} \quad T_{\text{obs}} \frac{b-c}{\kappa \delta},
\end{align}
respectively. This implies that the firing rate of an IF-TEM $F_R$ with parameters $b, \kappa$, and $\delta$ are upper and lower bounded as 
\begin{align}
    \frac{b-c}{\kappa \delta} \leq F_R \leq  \frac{b+c}{\kappa \delta}.
    \label{eq:firing_bounds}
\end{align}
Our goal is to recover a continuous-time signal $x(t)$ of the form of \eqref{eq:fri} below, from the time instances $\{t_n$\}.

\subsection{Problem Formulation}
\label{subsec:problem}
Like in previous work \cite{rudresh2020time}, we consider a $T$-periodic FRI signal of the form
\begin{equation}
    x(t) = \sum_{p\in\mathbb{Z}}\sum_{\ell=1}^L a_{\ell} h(t-\tau_{\ell}-p T),
    \label{eq:fri}
\end{equation}
where $h(t)$ is a known, real-valued pulse and the amplitudes and delays $\{(a_{\ell},\tau_{\ell})|\tau_{\ell} \in [0, T),a_{\ell}\in(0, a_{\max}] \}_{\ell=1}^L$ are unknown parameters. This signal model is ubiquitous in applications such as radar \cite{bar2014sub,bajwa2011identification,rudresh2017finite}, ultrasound \cite{tur2011innovation,mulleti2014ultrasound}, and more. In these applications, $h(t)$ denotes a known transmit pulse which is reflected from $L$ targets. The reflected signal is modeled as $x(t)$ where $a_\ell$ and $\tau_\ell$ denote the amplitude and time-delay corresponding to the $\ell$-th target.

In general, FRI signals can have wide bandwidth due to short duration pulses $h(t)$. However, by using the structure of the signal and knowledge of the pulse $h(t)$, FRI signals can be sampled at sub-Nyquist rates. This is typically achieved by passing $x(t)$ through a designed sampling kernel $g(t)$ and then measuring low-rate samples of the filtered signal $y(t)$ as shown in Fig. \ref{fig:FRIgeneric}. The kernel is designed such that the FRI parameters $\{a_\ell, \tau_\ell\}_{\ell=1}^L$ are computed accurately from the samples. In particular, it has been shown that $2L$ samples of $y(t)$ in an interval of length $T$, that are measured either uniformly \cite{tur2011innovation,vetterli2002sampling} or non-uniformly \cite{mulleti2016fri,donoho2006compressed}, are sufficient to determine $\{a_\ell, \tau_\ell\}_{\ell=1}^L$ uniquely. The reconstruction or determination of the parameters from the samples is achieved by applying spectral analysis methods such as the annihilating filter \cite{eldar2015sampling,vetterli2002sampling}.

As discussed in the introduction, a conventional FRI sampling scheme, such as in Fig.~\ref{fig:FRIgeneric} has a sampler which is controlled by a global clock that operates at the rate of innovation $\frac{2L}{T}$ Hz. For a large $L$ or a small $T$, the sampling rate increases, and the global clock requires high power. In this case, an IF-TEM sampler is well suited as it does not require a global clock.  


We consider the problem of perfect recovery of the FRI parameters $\{a_\ell, \tau_\ell\}_{\ell=1}^L$ by using an IF-TEM sampling scheme as shown in Fig.~\ref{fig:TEMfri}. Specifically, we consider designing the sampling kernel $g(t)$ and an IF-TEM such that the FRI parameters are uniquely determined from the time-encodings by keeping the firing rate close to the rate of innovation.

In the next section, we follow a similar strategy to the one in \cite{rudresh2020time}, and show that in the noiseless case, perfect recovery is guaranteed using as few as $2L+2$ firings in an interval of length $T$. 
In Section \ref{subsec:mthd2}, we suggest an alternative approach that is more robust to noise.

\section{FRI-TEM: sampling and perfect recovery of FRI Signals from IF-TEM measurements}
In this section, we show that FRI signals can be perfectly recovered from IF-TEM measurements. We use the fact that the FRI signal $x(t)$ in \eqref{eq:fri} can be perfectly reconstructed from its $2L$ FCSs. We derive conditions on the IF-TEM parameters and the sampling kernel $g(t)$ such that $2L$ FCSs of the input FRI signal can be uniquely recovered from the IF-TEM output. Our approach is similar to the one considered in \cite{rudresh2020time}. Specifically, up to \eqref{eq:mesRel}, our derivations are almost identical. However, in contrast to \cite{rudresh2020time}, we mathematically derive exact recovery guarantees.

\subsection{Fourier-Series Representation of FRI Signals}
\label{sub:Fourier}
We begin by explicitly relating the input signal $x(t)$ of \eqref{eq:fri} to its FSCs (cf. \eqref{eq:xbyh}).

Since $x(t)$ is a $T$-periodic signal, it has a Fourier series representation
\begin{equation}
    x(t) = \sum_{k\in\mathbb{Z}}\hat{x}[k]e^{jk\omega_0 t},
    \label{eq:x_initial}
\end{equation}
where $\omega_0=\frac{2\pi}{T}$. The Fourier-series coefficients $\hat{x}[k]$ are given by 
\begin{equation}
 \hat{x}[k] = \frac{1}{T} \hat{h}(k\omega_0)\sum_{\ell=1}^{L}a_{\ell}e^{-j k\omega_0\tau_{\ell}},
 \label{eq:x_fourier}
\end{equation}
where, $\hat{h}(\omega)$ is the continuous-time Fourier transform of $h(t)$ and we assume that $\hat{h}(k\omega_0) \neq 0$ for $k\in\mathcal{K}$  where $\mathcal{K}$ is a given set of indices. Since $x(t)$ is real-valued, its FSCs $\hat{x}[k]$ are complex conjugate pairs, that is,
\begin{align}
\hat{x}^*[-k] = \hat{x}[k].
\label{eq:complex_conj}
\end{align}

The sequence
\begin{equation}
    \frac{\hat{x}[k]}{\hat{h}(k\omega_0)}= \frac{1}{T}\sum_{\ell=1}^{L}a_{\ell}e^{-j k\omega_0\tau_{\ell}}, 
    \label{eq:xbyh}
\end{equation}
consists of a sum of $L$ complex exponentials. From the theory of high-resolution spectral estimation \cite{stoica}, it is well known that $2L$ consecutive samples of $\frac{\hat{x}[k]}{\hat{h}(k\omega_0)}$ are sufficient to determine $\{a_{\ell}, \tau_{\ell}\}_{\ell=1}^L$. For example, one can apply the well known annihilating filter method \cite{vetterli2002sampling} to compute $\{a_{\ell}, \tau_{\ell}\}_{\ell=1}^L$. In practice, the pulse $h(t)$ has short-duration and wide bandwidth. Hence, there always exist $2L$ or more non-vanishing Fourier samples $\hat{h}(k\omega_0)$ that are computed a priori. To determine the FRI parameters, we need to compute $2L$ consecutive values of $\hat{x}[k]$. Our problem is then reduced to that of uniquely determining the desired number of FSCs from the signal measurements. Since $x(t)$ typically consists of a large number of FSCs, we discuss next a sampling kernel design which removes unnecessary FSCs and thus reduces the sampling rate.
 
\vspace{-1mm}
\subsection{Sampling Kernel}
\label{sub:Kernel}
Since a minimum of $2L$ FSCs are sufficient for uniquely recovering the FRI signal, the sampling kernel $g(t)$ is designed to remove or annihilate any additional FSCs. The filtered signal $y(t)$ is given by
\begin{equation}
\begin{split}
    y(t) &= (x * g)(t)= \int_{-\infty}^{\infty} x(\tau)g(t-\tau)d\tau\\
    &=\sum_{k\in \mathbb{Z}}
    \hat{x}[k]\int_{-\infty}^{\infty} g(t-\tau)e^{jk\omega_0 \tau}d\tau\\
    &= \sum_{k\in \mathbb{Z}}
    \hat{x}[k]\, \hat{g}(k\omega_o) \,e^{jk\omega_0 t}.\\
    \end{split}
     \label{eq:yt_by_x}
\end{equation}
To restrict the summation to a finite number of terms and annihilate the unwanted FSCs we define the filter to satisfy the following condition in the Fourier domain:
\begin{equation}
    \hat{g}(k\omega_0) = 
    \begin{cases}
      1 & \text{if $k\in\mathcal{K}$},\\
      0 & \text{otherwise},
    \end{cases}
    \label{eq:kernel}
\end{equation}
where $\mathcal{K}$ is a set of integers such that $|\mathcal{K}|\geq 2L$.

One particular choice of the sampling kernel is a sum-of-sincs (SoS) kernel \cite{tur2011innovation} generated by
\begin{equation}
\hat{g}(\omega) = \sum_{k\in \mathcal{K}} \text{sinc}\left( \frac{\omega}{\omega_0}-k \right), 
\label{eq:SoStime}
\end{equation}
and 
\begin{align}
 g(t) = 
 \begin{cases}
   \displaystyle \sum_{k \in \mathcal{K}} e^{j k\omega_0 t}, \quad t \in (-T/2, T/2]\\
   0, \quad \text{elsewhere}.
 \end{cases}
 \label{eq:sos_ir}
\end{align}
The sampling kernel $g(t)$ is designed to pass the coefficients $\hat{x}[k], k\in\mathcal{K}$ while suppressing all other coefficients $\hat{x}[k], k\not\in\mathcal{K}$. 
Note that one can also apply an ideal lowpass filter with appropriate cutoff frequency to remove the FSCs. However, the impulse response of an ideal lowpass filter has infinite support, whereas the SoS kernel has compact support.

Using a SoS kernel, the filtered signal $y(t)$ is 
\begin{equation}
    y(t) =\sum_{k\in \mathcal{K}}
    \hat{x}[k]\hat{g}(k\omega_0)  e^{jk\omega_0 t}=\sum_{k\in \mathcal{K}}
    \hat{x}[k] e^{jk\omega_0 t}.
    \label{eq:yt_by_x22}
\end{equation}
The filtered signal $y(t)$ is sampled by an IF-TEM which requires its input to be real-valued and bounded. Since $\hat{x}[k]$ are conjugate symmetric, to ensure that $y(t)$ is real valued, the support set $\mathcal{K}$ is chosen to be symmetric around zero, that is, $\mathcal{K}$ is given as
\begin{equation}
    \mathcal{K} = \{-K,\cdots,K\},
    \label{eq:Kset}
\end{equation}
where $K \geq L$ to ensure that there are at-least $2L$ FSCs of $x(t)$ retained in $y(t)$.

From \eqref{eq:fri} and $y(t) = (x*g)(t)$, it can be shown that
\begin{align}
    c \triangleq \max\limits_t |y(t)| &\leq  L\,\, a_{\max} \,\, \|(h*g)\|_{\infty} \\
    &\leq L\,\,a_{\max} \,\,
    \|g\|_{\infty} \|h\|_{1}, 
   \label{eq:c}
\end{align}
where Young's convolution inequality is used. Since $|g(t)| \leq |\mathcal{K}|$ and $\mathcal{K}$ is a finite set, $g(t)$ is bounded. Hence $y(t)$ is bounded provided that $a_{\max} < \infty$ and the pulse $h(t)$ is absolutely integrable. In the remaining of the paper, we assume that both these conditions hold. 
\begin{figure*}[!h]
 \begin{align}
  \mathbf{A}=     \begin{bmatrix}
 e^{-jK\omega_0t_2}-e^{-jK\omega_0t_1}& \cdots & t_2 -t_1  & \cdots & e^{jK\omega_0t_2}-e^{jK\omega_0t_1}\\       
 e^{-jK\omega_0t_3}-e^{-jK\omega_0t_2}& \cdots & t_3 -t_2  & \cdots & e^{jK\omega_0t_3}-e^{jK\omega_0t_2}\\ 
 \vdots&  & \vdots &  & \vdots\\
e^{-jK\omega_0t_N}-e^{-jK\omega_0t_{N-1}}& \cdots & t_N -t_{N-1}  & \cdots & e^{jK\omega_0t_N}-e^{jK\omega_0t_{N-1}}\\  
\end{bmatrix}.
\label{eq:Amat}
  \end{align}
\end{figure*}
\vspace{-1cm}
\subsection{FRI TEM Sampling}
\label{sub:Samp}
The IF-TEM input is the filtered signal $y(t)$, which is the $T$-periodic signal defined in \eqref{eq:yt_by_x22}.
The output of the IF-TEM is a set of time instants $\{ t_n\}_{n\in\mathbb{Z}}$. Given $\{t_n\}$ one can determine the measurements $\{y_n\}$ by using \eqref{eq:trigger0}.
The relation between the measurements $y_n$ and the desired FSCs is given by
\begin{equation}
    \begin{split}
    y_n &= \int_{t_n}^{t_{n+1}}y(s)\, ds\\ &= \int_{t_n}^{t_{n+1}} \sum_{k\in \mathcal{K} \setminus \{ 0\} }
    \hat{x}[k]e^{jk\omega_0 t}dt+ \int_{t_n}^{t_{n+1}}\hat{x}[0]dt\\
    &= \sum_{k\in \mathcal{K} \setminus \{ 0\} } \hat{x}[k] \frac{\left( e^{jk\omega_0 t_{n+1}}- e^{jk\omega_0 t_{n}} \right)}{jk\omega_0}+\hat{x}[0] \left( t_{n+1} - t_{n}\right).
    \end{split}
    \label{eq:yx_rel}
\end{equation}

To extract the desired FSCs from \eqref{eq:yx_rel}, we denote by $\mathbf{y}$ the vector $[\int_{t_1}^{t_2}y(t)dt, \int_{t_2}^{t_3}y(t)dt, \cdots, \int_{t_{N-1}}^{t_N}y(t)dt]^{\top}$, where $N$ is the number of time instants in the interval $T$.
In addition, 
\begin{align}
\mathbf{\hat{x}} = \left[-\frac{\hat{x}[-K]}{jK\omega_0}, \cdots, \hat{x}[0],\cdots, 
\frac{\hat{x}[K]}{jK\omega_0}\right]^{\top}.
\label{eq:xhat}
\end{align}
Let, the matrix $\mathbf{A}$ be given in \eqref{eq:Amat}. With this notation, \eqref{eq:yx_rel} can be written in the following matrix form:
\begin{equation}
    \mathbf{y}=\mathbf{A}\mathbf{\hat{x}}.
    \label{eq:mesRel}
\end{equation}
This equation describes the relation between the IF-TEM measurements and the FSCs. Our goal is to determine these FSCs embedded in $\mathbf{\hat{x}}$. As previously mentioned, we can perfectly recover the FRI parameters from $\mathbf{\hat{x}}$. Thus, a natural question to ask is under what conditions the matrix $\mathbf{A}$ has a unique left-inverse, that is, the matrix $\mathbf{A}$ has full column rank. If invertibility is guaranteed, then the Fourier coefficients vector can be computed as  
\begin{equation}
    \mathbf{\hat{x}} =\mathbf{A}^{\dagger} \mathbf{y},
    \label{eq:inverseFCS}
\end{equation}
where $\mathbf{A}^\dagger$ denotes the Moore-Penrose inverse of $\mathbf{A}$ in \eqref{eq:Amat}.

In \cite{rudresh2020time}, the matrix is assumed to be uniquely left-invertible. The authors showed via simulations that the matrix has full column rank, however, a concrete proof is not presented. In the following we show that for an adequate number of firings, the matrix is uniquely invertible.

\subsection{Recovery Guarantees}
\label{sub:REC}
In this section, we present our main results where we show that for the sampling kernel choice \eqref{eq:sos_ir}, we can uniquely identify the FSCs from the IF-TEM time instants. Specifically, we show that for a particular choice of the IF-TEM parameters, the matrix $\mathbf{A}$ defined in \eqref{eq:Amat} is left invertible.
Our results are summarized in the following theorems.

\begin{theorem}
\label{theorem:Ainv}
Consider a positive integer $K$ and a number $T>0$. Let $0\leq t_1<t_2<\cdots<t_N < T$ for an integer $N$, and $\omega_0 = \frac{2\pi}{T}$. Then the matrix $\mathbf{A}$ defined in \eqref{eq:Amat} is left-invertible provided that $N \geq 2K+2$.
\end{theorem}
\begin{proof}
See the Appendix (cf Case-2).
\end{proof}

Theorem~\ref{theorem:Ainv} implies that there should be a minimum of $2K+2$ IF-TEM time instants within an interval of $T$ to enable recovery of the FSCs, and subsequently reconstruction of the FRI signal. To ensure this, the minimum firing rate $\frac{b-c}{\kappa\delta}$ (cf. \eqref{eq:firing_bounds}) should be chosen such that 
\begin{align}
    \frac{b-c}{\kappa\delta} \geq \frac{2K+2}{T}.
    \label{eq:minFR}
\end{align}

By combining Theorem \ref{theorem:Ainv}, the result in \eqref{eq:minFR}, and the fact that $2L$ FSCs are sufficient to recover the FRI parameters, we summarize the sampling and reconstructions of FRI signals using IF-TEM in the following theorem.

\begin{theorem}
\label{theorem:FRI0}
Let $x(t)$ be a $T$-periodic FRI signal of the following form
\begin{align*}
      x(t) = \sum_{p\in\mathbb{Z}}\sum_{\ell=1}^L a_{\ell} h(t-\tau_{\ell}-p T),
\end{align*}
where $t_\ell \in [0, T)$, $|a_\ell|<\infty$, and $L$ is known. We assume that the amplitudes $\{a_\ell\}_{\ell=1}^L$ are finite, and the pulse $h(t)$ is known and absolutely integrable. Consider the sampling mechanism shown in Fig..~\ref{fig:TEMfri}. Let the sampling kernel $g(t)$ satisfy
\begin{align*}
    \hat{g}(k\omega_0) =       
    \begin{cases}
      1 & \text{if $k\in \mathcal{K}=\{-K,\cdots,K \}$},\\
      0 & \text{otherwise},
    \end{cases}
    \end{align*}
and $\max\limits_t|(h*g)(t)|<\infty$. Choose the real positive TEM parameters $\{b, \kappa, \delta\}$ such that $c<b<\infty$, where $c$ is defined in \eqref{eq:c}, and
\begin{equation}           
\frac{b-c}{\kappa\delta} \geq \frac{2K+2}{T}. \label{eq:sample_bound0}
\end{equation} 
Then. the parameters $\{a_{\ell}, \tau_{\ell}\}_{\ell=1}^L$ can be perfectly recovered from the TEM outputs if $K \geq L$.
\end{theorem}

Based on Theorem~\ref{theorem:FRI0}, a reconstruction algorithm to compute the FRI parameters from TEM firings is presented in Algorithm~\ref{alg:algorithm_th1}.

\subsection{IF-TEM Parameter Selection}
\label{subsec:params}
The IF-TEM parameters are selected such that there is a minimum of $N\geq 2L+2$ time instants $\{t_n\}_{n=1}^N$ within a time interval $T$.
Thus, the minimum firing rate that enables accurate reconstruction is $\frac{2L+2}{T}$. The maximum firing rate is bounded by $\frac{b+c}{\kappa\delta}$.
While the threshold $\delta$, which is a parameter of the comparator, is easier to control, the integrator constant $\kappa$ is a parameter of the integrator, and it is usually fixed. Thus, assuming a fixed value of $b$ and $\kappa$, choosing small $\delta$ results in a large firing rate above the minimum desirable value of $\frac{2L+2}{T}$.
In practice, both $b$ and $\delta$ are generated through a DC voltage source, and therefore large values of bias and threshold require high power. Hence, to minimize the power requirements, it is desirable for $b$ and $\delta$ to be as small as possible.

\subsection{Simulations}
\label{sub:sim}
We next numerically validate Theorem \ref{theorem:FRI0}. In Fig. \ref{fig:deltas}, we consider $h(t)$ as a Dirac impulse with time period $T=1$ seconds. We consider the simulations for $L = 3$, $5$, and $10$. The time delays and amplitudes are selected uniformly at random over $(0, 1)$. The input signal $x(t)$ is filtered using an SoS sampling kernel with $\mathcal{K} = \{-K,\cdots,K\}$, where $K=L$. The filtered output $y(t)$ is sampled using an IF-TEM which has a threshold $\delta = 0.07$ and $\kappa=1$. The bias  of the IF-TEM is set as $b=0.9$, $1.3$, and $2.5$ for $L = 3$, $5$, and $10$ respectively.
The parameters are chosen to satisfy the inequality in \eqref{eq:sample_bound0}, and resulted in $13$, $18$, and $36$ samples per period for $L = 3$, $5$, and $10$. As per Theorem \ref{theorem:FRI0}, $8$, $12$, and $22$ samples per period are sufficient. 
The reconstruction was found to be stable even for a larger number of impulses.
We summarize the choice of IF-TEM parameters and the resulting firing rate in Table \ref{Tab:Tcr}.
In Fig. \ref{fig:spline}, using the same filter, we depict the estimation of a stream of pulses with $L=3$, $\{a_{\ell}\}= \{0.5, -0.45, 0.4\}$, and $\{ \tau_{\ell}\} = \{0.2, 0.33, 0.8\}$. The IF-TEM $b,\delta,\kappa$ which satisfy the inequality in \eqref{eq:sample_bound0}, are $0.9,0.07,1$ respectively. The resulting firing rate is as few as 13 samples/s.
\begin{table}[ht]
\caption{IF-TEM Parameters Choice for Estimation}
\begin{center}
\begin{tabular}{ |c|c|c|c|c|c| } 
 \hline
 $L$ & $b$ & $\delta$ & $\kappa$ & $\omega_0$ & $F_R$ (samples/s)\\
 \hline
 3 & 0.9 & 0.07 & 1 & $2\pi$ & 13 \\
 \hline
 5 & 1.3 & 0.07 & 1 & $2\pi$ & 18 \\
  \hline
 10 & 2.5 & 0.07 & 1 & $2\pi$ & 36 \\

    \hline
    \end{tabular}
\end{center}
\label{Tab:Tcr}
\end{table}

	\begin{figure*}[hbt!]
		\centering
		\includegraphics[width=1\textwidth]{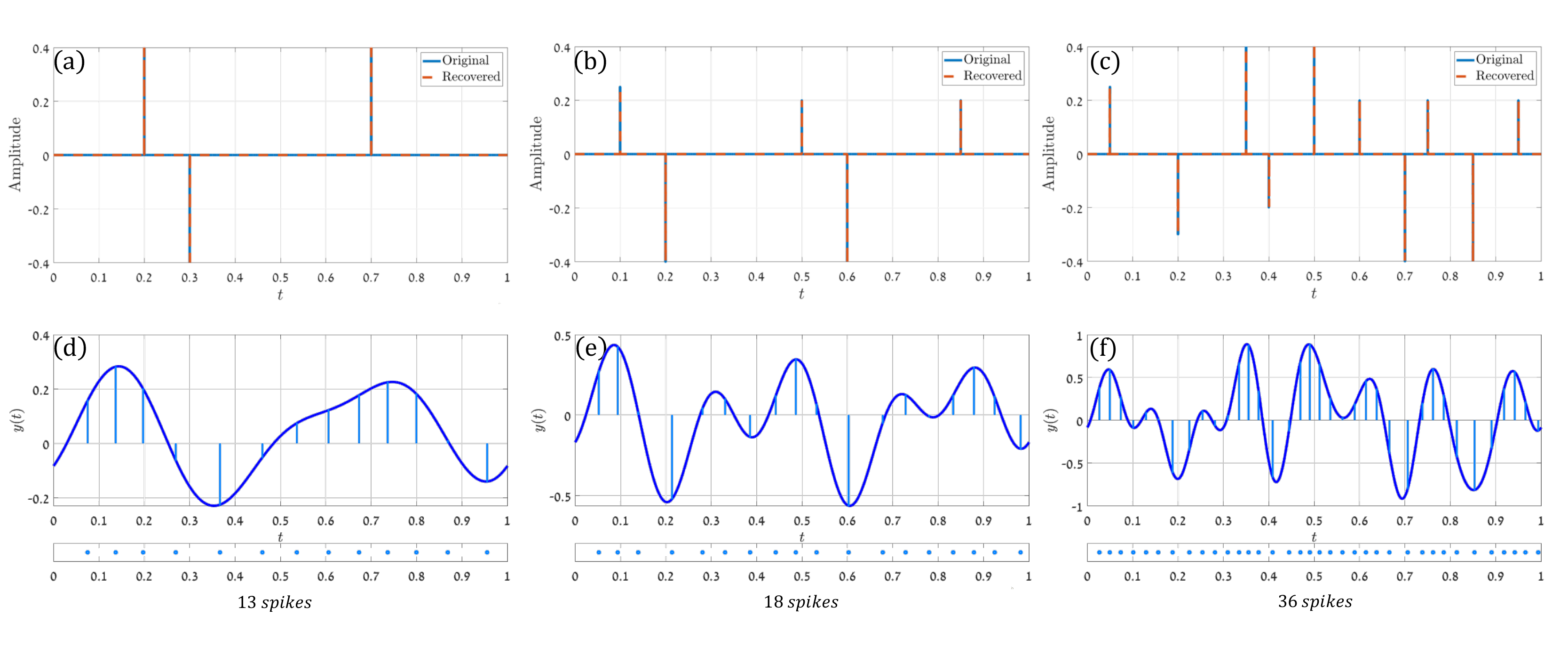}
		\caption{Sampling and reconstruction of stream of Dirac impulses using TEM by applying the SoS kernel. Both $\tau_{\ell}$ and $a_{\ell}$ are chosen uniformly at random over $(0, 1)$. (a)-(c): the input signal and its reconstruction for $L=3$, $L=5$, and $L=10$ respectively. (e)-(g): the filtered signal $y(t)$ and the time instants $t_n$ for $L=3$, $L=5$, and $L=10$ respectively.
 }
		\label{fig:deltas}
	\end{figure*}
	
		\begin{figure}[t!]
		\centering
		\includegraphics[width= 3
		in]{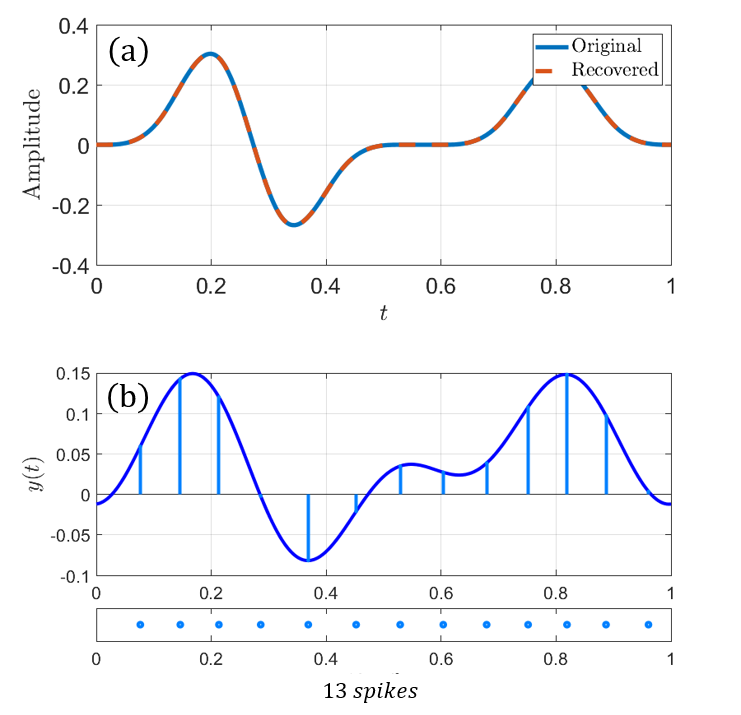}
        \caption{Sampling and reconstruction of stream of pulses using TEM by applying the SoS kernel. (a): the input signal and its reconstruction for $L=3$. (b): the filtered signal $y(t)$ and the time instants $t_n$ for $L=3$.}
		\label{fig:spline}
	\end{figure}
	
\begin{algorithm}[!h]
\begin{algorithmic}[1]
\caption{Reconstruction of a $T$-periodic FRI signal using Theorem \ref{theorem:FRI0}.}
\label{alg:algorithm_th1}
\Statex \textbf{Input:} $N\geq 2K+2$ spike times $\{t_n\}_{n=1}^N$ in a period $T$.
\State Let $n \leftarrow 1$
\Loop
        \While {$n<N-1$}  
             \State   Compute  
            $y_n = -b(t_{n+1}-t_n)+\kappa\delta$
            \State $n:=n+1$.
\EndWhile
\State Compute Fourier coefficients vector $\mathbf{\hat{x}} = \mathbf{A}^\dagger\,\mathbf{y}$ in \eqref{eq:inverseFCS}.
    \State   Estimate $\{(a_{\ell},\tau_{\ell})_{\ell=1}^L$ using spectral
    analysis method.
\EndLoop 
\Statex \textbf{Output:}
$\{(a_{\ell},\tau_{\ell})_{\ell=1}^L$.
\end{algorithmic}
\end{algorithm}


\section{Noise Robustness}
While the results in \cite{rudresh2020time} and the previous section assume that there is no measurement noise, in practice, the signals are contaminated by noise. In the presence of noise, the IF-TEM outputs or time instants are perturbed. While using Algorithm~\ref{alg:algorithm_th1}, this results in a perturbation in the matrix $\mathbf{A}$ as well as the measurements $\mathbf{y}$ in \eqref{eq:mesRel}. In this case, when computing the FSCs using \eqref{eq:mesRel}, the stability of $\mathbf{A}$, which is measured by the condition number of the matrix, impacts the results.
Next, we show that by excluding zero from $\mathcal{K}$, perfect recovery is possible, and in the noisy scenario, the method is more robust.

As shown in Appendix, the matrix $\mathbf{A}$ and $\mathbf{B}$ have full column ranks provided that vector $\tilde{x}[0] \mathbf{t}$ is not in the column space of matrix $\mathbf{W}$ (cf. (41)). The condition holds for $N >2K+1$. Note that $\bar{x}[0] \neq 0$ when we consider matrix $\mathbf{A}$ and $\bar{x}[0] = 0$ in the context of matrix $\mathbf{B}$. For $N > 2K+1$ and $\bar{x}[0] \neq 0$, it is possible that there exist a set of time encodings $\{t_n\}_{n=1}^N$ such that vector $\bar{x}[0] \mathbf{t}$ comes arbitrary close to $\mathbf{W}$'s column space. It implies that the straight line $r(t)$ is close to the trigonometric polynomial $f(t)$ at $\{t_n\}_{n=1}^N$.
\begin{figure}[!h]
	\centering
	\subfigure[Condition numbers of matrices $\mathbf{A}$ is 30 and $\mathbf{B}$ is 3]{\includegraphics[width=2 in]{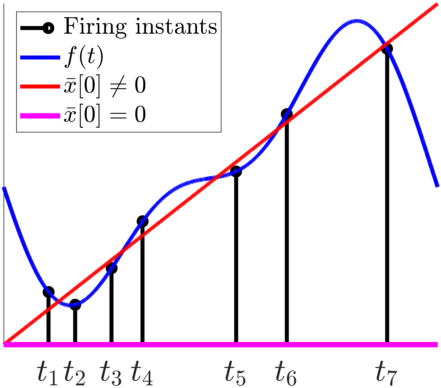}}
		\subfigure[Condition numbers of matrices $\mathbf{A}$ is 3000 and $\mathbf{B}$ is 5]{\includegraphics[width=2.1 in]{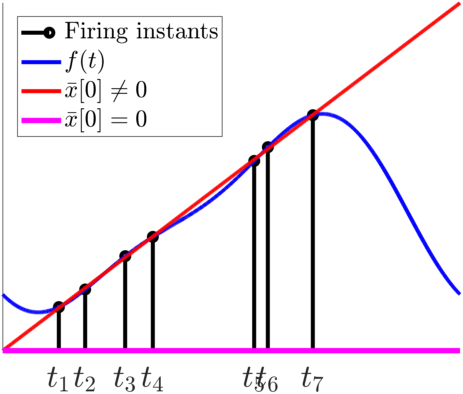}}
	\caption{Examples $f(t)$ and $r(t)$ when the matrix $\mathbf{A}$ has large condition number.}
	\label{fig:cond_bad}
\end{figure}
In such case, from (41), we can determine a set of Fourier coefficients $\{\bar{x}[k]\}_{k=-K}^K$ such that $\mathbf{Vx}\approx c \mathbf{1}_N$ and consequently $\mathbf{Ax}\approx  \mathbf{0}_N$, where $\mathbf{0}_N$ is a zero vector of length $N$. Hence for that particular set of $\{t_n\}_{n=1}^N$, $\mathbf{A}$ becomes ill-conditioned. However, when $\bar{x}[0] = 0$, we observe that it is less likely that $r(t)$ with zero-slope becomes closer to $f(t)$ at time-encoding instants. A couple of illustrative examples depicting the intuition are shown in Fig~\ref{fig:cond_bad} where $K = 2$ and $N = 7$. In Fig.~\ref{fig:cond_bad}(a), the condition number of matrices $\mathbf{A}$ and $\mathbf{B}$ are $30$ and $3$, respectively. Although, the condition numbers are small, we note that the condition number of $\mathbf{A}$ is ten times higher than that of $\mathbf{B}$. This is because the straight line correspond to $\mathbf{A}$ (shown in red) is closer to the trigonometric polynomial at the time-encodings than that of $\mathbf{B}$ (in magenta). In the example shown in Fig.~\ref{fig:cond_bad}(b), the condition number of matrices $\mathbf{A}$ is 3000 as the $|r(t) - f(t)|$ is small for $t \in \{t_n\}_{n=1}^N$. Whereas, $|r(t) - f(t)|$ is relatively large for $\bar{x}[0] = 0$ and consequently $\mathbf{B}$ has lower condition number. In other words, it is less likely that a zero vector comes closer to the range space of matrix $\mathbf{W}$ compared to a vector $\mathbf{t}$. Hence, matrix $\mathbf{B}$ has a better condition number compared to that of matrix $\mathbf{A}$.    
\subsection{Exclusion of Zero}
\label{subsec:mthd2}
In this section, we show that by excluding the zero frequency in $\mathcal{K}$ we achieve perfect reconstruction for FRI signals of the form of \eqref{eq:fri}.
In this case, the resulting matrix has a much more stable structure compared to the matrix $\mathbf{A}$ of \eqref{eq:Amat}. 
Suppose that we remove $k=0$, following \eqref{eq:yx_rel}, we have
\begin{equation}
    \begin{split}
    y_n &= \int_{t_n}^{t_{n+1}}y(s)\, ds\\ &= \int_{t_n}^{t_{n+1}} \sum_{k\in \mathcal{K} \setminus \{ 0\} }
    \hat{x}[k]e^{jk\omega_0 t}dt\\
    &= \sum_{k\in \mathcal{K} \setminus \{ 0\} } \hat{x}[k] \frac{\left( e^{jk\omega_0 t_{n+1}}- e^{jk\omega_0 t_{n}} \right)}{jk\omega_0}.
    \end{split}
    \label{eq:yx_relProof2}
\end{equation}

To extract the designed FSCs from \eqref{eq:yx_relProof2}, we denote by $\mathbf{y}_0$ the vector $[\int_{t_1}^{t_2}y(t)dt, \int_{t_2}^{t_3}y(t)dt, \cdots, \int_{t_{N-1}}^{t_N}y(t)dt]^{\top}$, where $N$ is the number of time instants in the interval $T$.
The measurements $\mathbf{y}_0$ and the FSCs
\begin{align}
\mathbf{\hat{x}}_0 = \left[-\frac{\hat{x}[-K]}{jK\omega_0}, \cdots, -\frac{\hat{x}[-1]}{j\omega_0},\frac{\hat{x}[1]}{j\omega_0},\cdots, 
\frac{\hat{x}[K]}{jK\omega_0}\right]^{\top}
\label{eq:xhatProof2}
\end{align}
are now related as 
\begin{equation}
    \mathbf{y}_0=\mathbf{B}\mathbf{\hat{x}}_0,
    \label{eq:mesRel2}
\end{equation}
where $\mathbf{B}$ is given as in \eqref{eq:Bmat}. Next, we show that the matrix $\mathbf{B}$ has full column rank and is uniquely left invertible.

\begin{figure*}[!h]
 \begin{align}
  \mathbf{B}=     \begin{bmatrix}
 e^{-jK\omega_0t_2}-e^{-jK\omega_0t_1}& \cdots & e^{-j\omega_0t_2}-e^{-j\omega_0t_1}&e^{j\omega_0t_2}-e^{j\omega_0t_1}  & \cdots & e^{jK\omega_0t_2}-e^{jK\omega_0t_1}\\       
 e^{-jK\omega_0t_3}-e^{-jK\omega_0t_2}& \cdots & e^{-j\omega_0t_3}-e^{-j\omega_0t_2}&e^{j\omega_0t_3}-e^{j\omega_0t_2}  & \cdots & e^{jK\omega_0t_3}-e^{jK\omega_0t_2}\\ 
 \vdots&  & \vdots & \vdots&  & \vdots\\
e^{-jK\omega_0t_N}-e^{-jK\omega_0t_{N-1}}& \cdots &e^{-j\omega_0t_N}-e^{-j\omega_0t_{N-1}}&e^{j\omega_0t_N}-e^{j\omega_0t_{N-1}}  & \cdots & e^{jK\omega_0t_N}-e^{jK\omega_0t_{N-1}}\\  
\end{bmatrix}.
\label{eq:Bmat}
  \end{align}
\end{figure*}

\begin{theorem}
\label{theorem:Binv}
Consider a positive integer $K$ and a number $T>0$. Let $0\leq t_1<t_2<\cdots<t_N < T$ for an integer $N$, and $\omega_0 = \frac{2\pi}{T}$. Then the matrix $\mathbf{B}$ defined in \eqref{eq:Bmat} is left-invertible provided that $N \geq 2K+2$.
\end{theorem}
\begin{proof}
The proof follows the same line as that of Theorem~\ref{theorem:Ainv} with the constraint $\hat{x}[0] = 0$ as detailed in Case-1 in the Appendix.
\end{proof} 
\noindent Since the left-inverse of $\mathbf{B}$ exists, the Fourier coefficients vector is computed as  
\begin{equation}
    \mathbf{\hat{x}}_0 =\mathbf{B}^{\dagger} \mathbf{y}_0.
    \label{eq:inverseFCSB}
\end{equation}

Although, the FSCs are computed uniquely, they are not consecutive unlike the FSCs computed in Theorem~\ref{theorem:FRI0}. Since high resolution spectral estimation techniques such as annihilating filter requires $2L$ consecutive FSCs, to uniquely determine the FRI parameters, we need $K \geq 2L$. This results in twice the firing rate compared to that in Theorem~\ref{theorem:FRI0}. An alternative approach to reduce the firing rate is to assume that the time-delays are on a grid. In this case, determination of time-delays and amplitudes of the FRI signal from FCSs is cast as a compressive sensing problem \cite[Section V-B]{bar2014sub}. This problem is efficiently solved from $2L$ FSCs, that are not necessarily consecutive, by using sparse recovery approaches such as orthogonal matching pursuit (OMP) \cite[Ch. 11]{eldar2015sampling}. Hence, by assuming that the time-delays of the FRI signal are on a grid, we require $K \geq L$.

The minimum firing rate for the IF-TEM is 
\begin{align}
    \frac{b-c}{\kappa\delta} \geq \frac{2K+2}{T},
    \label{eq:minFR2}
\end{align}
where $K \geq 2L$ for off-grid time-delays and $K \geq L$ for time-delays on-grid.

By combining Theorem \ref{theorem:Binv} with the result in \eqref{eq:minFR2}, we summarize the sampling and reconstruction of FRI signals using IF-TEM in the following theorem.

\begin{theorem}
\label{theorem:FRI}
Let $x(t)$ be a $T$-periodic FRI signal of the following form
\begin{align*}
      x(t) = \sum_{p\in\mathbb{Z}}\sum_{\ell=1}^L a_{\ell} h(t-\tau_{\ell}-p T),
\end{align*}
where the number of FRI signals $L$ is known, and $h(t)$ is a signal with known pulse shape. Consider the sampling mechanism shown in Fig.~\ref{fig:TEMfri}. Let the sampling kernel $g(t)$ satisfy
\begin{align*}
    \hat{g}(k\omega_0) =       
    \begin{cases}
      1 & \text{if $k\in \mathcal{K}=\{-K,\cdots,-1,1,\cdots,K \}$},\\
      0 & \text{otherwise},
    \end{cases}
    \end{align*}
 and $\max\limits_t|(h*g)(t)|<\infty$. Choose the real positive TEM parameters $\{b, \kappa, \delta\}$ such that $c<b<\infty$, where $c$ is defined in \eqref{eq:c}, and
\begin{equation}           
\frac{b-c}{\kappa\delta} \geq \frac{2K+2}{T}. \label{eq:sample_bound}
\end{equation} 
Then, the parameters $\{a_{\ell}, \tau_{\ell}\}_{\ell=1}^L$ can be perfectly recovered from the TEM outputs if
\begin{enumerate}
    \item $K \geq 2L$ when $\{t_\ell\}_{\ell=1}^L$ are off-grid
    \item $K \geq L$ when $\{t_\ell\}_{\ell=1}^L$ are on-grid.
\end{enumerate}
\end{theorem}

An algorithm to perfectly recover the FRI parameters from IF-TEM samples is summarized in  Algorithm \ref{alg:algorithm_grid}.
\begin{algorithm}[h]
\begin{algorithmic}[1]
\caption{ Reconstruction of a $T$-periodic FRI signal.}
\label{alg:algorithm_grid}
\Statex \textbf{Input:} $N\geq 2K+2$ spike times $\{t_n\}_{n=1}^N$ in a period $T$.
\Statex \textbf{Init:} $n:=1$.
\Loop
        \While {$n<N-1$}  
             \State   Compute  
            $y_n = -b(t_{n+1}-t_n)+\kappa\delta$
            \State $n:=n+1$.
        \EndWhile
    \State Compute Fourier coefficients vector $\mathbf{\hat{x}}_0 = \mathbf{B}^\dagger\,\mathbf{y_0}$ in \eqref{eq:inverseFCSB}. 
    \State   Estimate $\{(a_{\ell},\tau_{\ell})_{\ell=1}^L$ using a spectral analysis method, if the delays are on a grid then $K\geq L$, or $K \geq 2L$ otherwise.
\EndLoop 
\Statex \textbf{Output:} $\{(a_{\ell},\tau_{\ell})_{\ell=1}^L$.
\end{algorithmic}
\end{algorithm}

\subsection{Numerical Results}
\label{subsec:mthdcompare}
In many practical systems, the time instants can only be recorded with finite precision, i.e., in practical circumstances, the recorded times are effective time instances $\{t'_n\}$ which differ from the real-time instances $\{t_n\}$, and perfect reconstruction may no longer be possible \cite{gontier2014sampling,alexandru2019reconstructing}. 

We compare the robustness of the two Algorithms \ref{alg:algorithm_th1} and \ref{alg:algorithm_grid} in the presence of perturbation to the measured time instants. 
We demonstrate that the algorithm which uses the recovery method using the sampling kernel presented in Theorem \ref{theorem:FRI} gives better recovery in the presence of noise than Algorithm \ref{alg:algorithm_th1} which uses the sampling kernel presented in Theorem \ref{theorem:FRI0}.

In the above Algorithms, the first step is to estimate the Fourier samples from the TEM measurements  by taking pseudo inverses of matrix $\mathbf{A}$ (cf. \eqref{eq:mesRel}) and $\mathbf{B}$ (cf. \eqref{eq:mesRel2}), respectively. Both the matrices are functions of the measured time instants and the sampling kernel. In Fig.~\ref{fig:condnum}, we compare the condition numbers of the matrices with perturbed firing instants as a function of the number of FRI signals $L$. To that aim, 5000 random sets of monotonic sequences $\{t_n\in  [0,T)\}_{n=1}^N$ were used.
As shown in Fig. \ref{fig:condnum}, the condition number of the matrix $\mathbf{B}$ is substantially smaller than the condition number of the matrix $\mathbf{A}$. 

\begin{figure}[t!]
	\centering
	\includegraphics[width= 3.7 in]{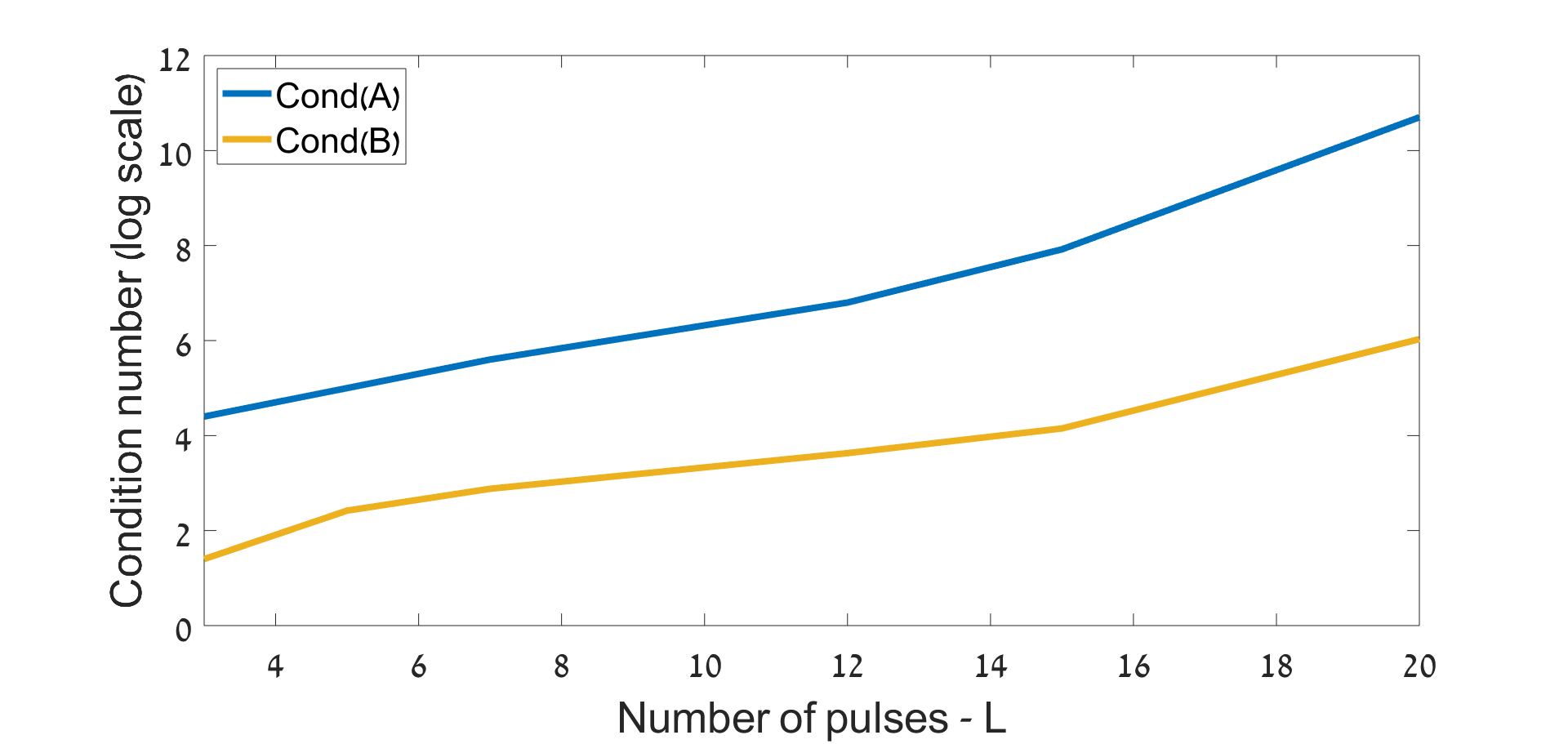}
	\caption{Average condition number of the matrices $\mathbf{A}$ and $\mathbf{B}$ as a function of $L$.}
		\label{fig:condnum}
\end{figure}

Next, we illustrate the reconstruction of a $T$-periodic FRI signal from non-uniform noisy samples (time instants) using the two reconstruction algorithms of Theorems \ref{theorem:FRI0} and \ref{theorem:FRI}.
We created a periodic FRI signal $x(t)$ of the form of \eqref{eq:fri}. The signal $x(t)$ with period $T=1$ consists of $L=3$ pulses with $h(t) = \beta^{(3)}(20t)$, where $\beta^{(3)}(t)$ is a third-order cubic B-spline, with 
$\{a_{\ell}\}_{\ell=1}^L$= \{0.5, -0.45, 0.4\}, and $\{\tau_{\ell}\}_{\ell=1}^3 = \{0.2, 0.4,  0.8\}$. 
The TEM parameters are $b = 1.2$, $\kappa=1$, and $\delta$ changes from $0.04$ to $0.09$ resulting in $13$ to $24$ time samples. The parameters are chosen to satisfy condition \eqref{eq:sample_bound}.
We consider a sum-of-sincs kernel with $\mathcal{K} =\{-K, \cdots, 0, \cdots, K\}$ for Theorem \ref{theorem:FRI0}, and $\mathcal{K} =\{-K, \cdots, -1, 1, \cdots, K\}$ for Theorem \ref{theorem:FRI}.
For both kernels, the time instances $\{t_n\}$ were perturbed
by a zero-mean white Gaussian noise with mean $0$ and variance $0.001$. 
The FRI parameters are computed by applying OMP \cite{eldar2015sampling}.

Reconstruction accuracy of the two algorithms is compared in terms of relative mean square error (MSE), given by
\begin{equation}
    \text{MSE} = \frac{||{x(t)-\Bar{x}(t)}||_{L_2[0,T]}}{||{x(t)}||_{L_2[0,T]}},
\end{equation}
where $\Bar{x}(t)$ is the reconstructed signal. In Fig.~\ref{fig:MSE}, we show the MSE of the two algorithms as a function of the number of noisy time instances. The MSE of Algorithm~\ref{alg:algorithm_grid} is 2-6 dB lower compared to that of Algorithm~\ref{alg:algorithm_th1} for different firing rates.

	\begin{figure}[hbt!]
	\hspace*{-0.6cm}
		\centering
	\includegraphics[width=0.56\textwidth]{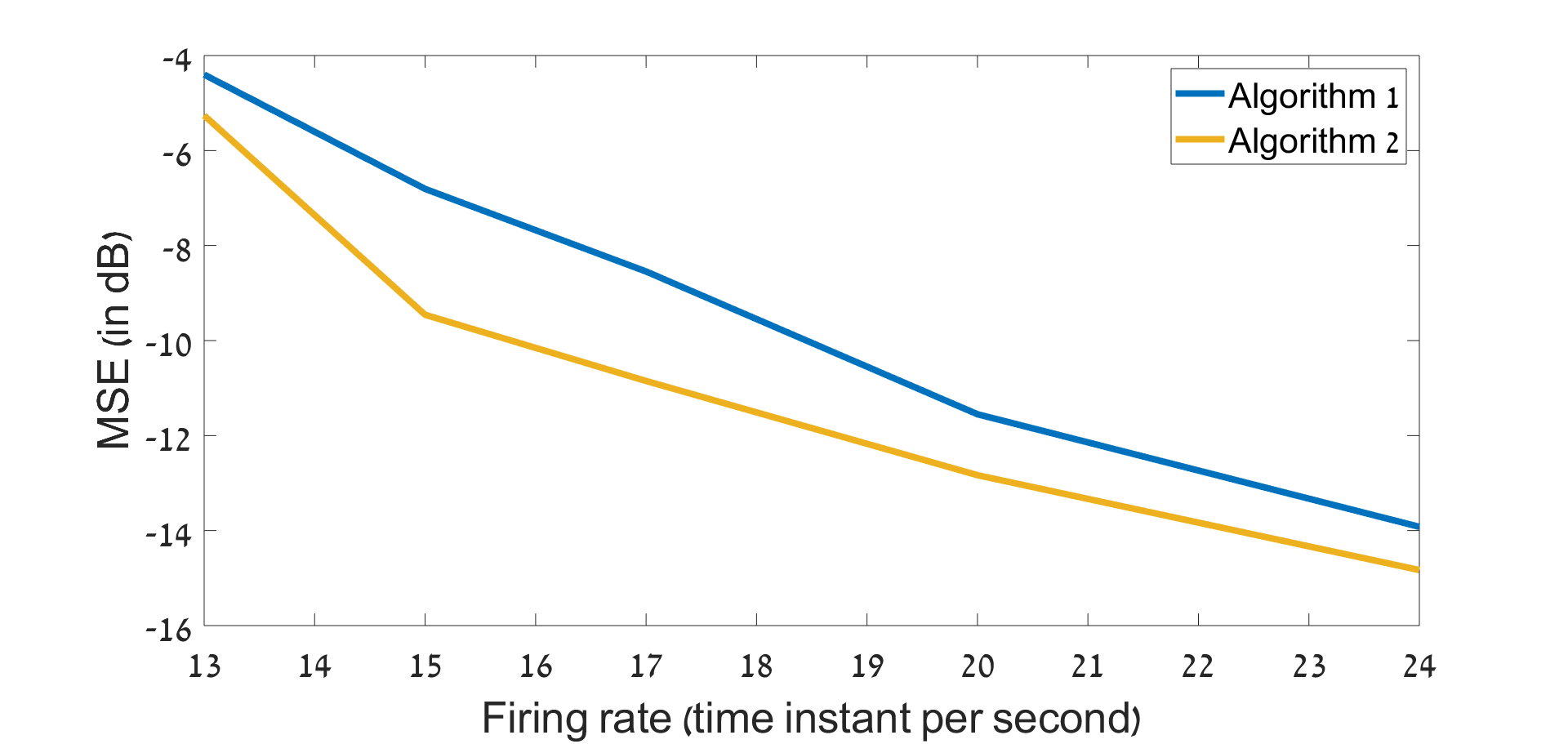}
		\caption{Performance comparison Algorithm~\ref{alg:algorithm_th1} and Algorithm~\ref{alg:algorithm_grid} in the presence
        of continuous-time white Gaussian noise with zero mean and variance 0.001.}
		\label{fig:MSE}
	\end{figure}
\section{Reconstruction of Nonperiodic FRI Signals}
\label{sec:Nonperiodic}

Consider a nonperiodic FRI signal of the form
\begin{equation}
    \Tilde{x}(t) = \sum_{\ell=1}^L a_{\ell} h(t-\tau_{\ell}),
    \label{eq:nonperiodic}
\end{equation}
where $h(t)$ is a known pulse and the amplitudes and delays $\{(a_{\ell},\tau_{\ell})|\tau_{\ell} \in [0, T),a_{\ell}\in\mathbb{R} \}_{\ell=1}^L$ are unknown parameters. We assume that the pulse $h(t)$ has finite support $R$, namely 
\begin{equation}
    h(t) = 0, \quad \forall|t|\geq \frac{R}{2}.
\end{equation}
Given that our main interest is in pulses with a very wide or even infinite spectrum, i.e very short pulses, traditional sampling techniques will prove to be ineffective in our case \cite{tur2011innovation}.
We design a sampling kernel $\Tilde{g}(t)$ such that
\begin{equation}
    \Tilde{y}(t) = \Tilde{x}(t)*\Tilde{g}(t) = {y}(t),\quad  \forall t\in [0,T),
\end{equation}
where $y(t)$ defined in \eqref{eq:yt_by_x}. Specifically, $\Tilde{g}(t)$ is compactly supported and defined by 
\begin{equation}
    \Tilde{g}(t) = \sum_{s=-S}^S g(t+sT),
\end{equation}
where $S$ is determined by $R$ and $T$ (More details are available in \cite{tur2011innovation}).
Since both time instants and the time delays of the FRI signals are within the interval $[0,T)$, i.e., $t_n\in[0,T)$ and $\tau_\ell \in [0,T)$, $\ell = 1, \cdots, L$, the time instants taken in the nonperiodic case using IF-TEM are the same as in the periodic case. 
Therefore, the recovery guarantees developed for non-periodic case in Theorem \ref{theorem:FRI} and Theorem~\ref{theorem:FRI0} are applicable to non-periodic case.


\section{Conclusions}
\label{sec:Conclusions}
In this paper, we developed two sampling and reconstruction frameworks for periodic FRI signals by using IF-TEMs. We showed that prefect recovery is achieved by both methods. 
While the first method operates close to the rate of innovation, the second requires a higher firing rate. However, in the presence of noise, the second approach is more robust. Our claims are supported by simulation results.
Compared to conventional amplitude-based sampling for FRI signals the proposed TEM-based method is less power consuming and hence, more cost effective.


\appendix
\section{FRI-TEM Proof}
In this section we present proof for Theorem~\ref{theorem:Ainv} and Theorem~\ref{theorem:Binv}. We consider a unified approach to prove both the theorems by using the fact that the proof of Theorem~\ref{theorem:Binv} is a special case of that of Theorem~\ref{theorem:Ainv} with $\hat{x}[0]=0$. 
\label{Proof1}
\begin{proof}
The matrix $\mathbf{A}$ in \eqref{eq:mesRel} is decomposed as 
\begin{align}
\mathbf{A} = \mathbf{DV},
\label{eq:A=DV}
\end{align}
where 
\begin{align}
\mathbf{D} = 
\begin{pmatrix}
-1 & 1 & 0 &\cdots & 0 \\
0 & -1 & 1 &\cdots & 0\\
0 & 0 &-1 & \cdots & 0 \\
\vdots & \vdots & \vdots & \ddots& \vdots \\
0 &0 &0 & \cdots &1
\end{pmatrix} \in \mathbb{R}^{(N-1) \times N}
\label{eq:dmat}
\end{align}
and $\mathbf{V}$ is given as in \eqref{eq:vmat}. To determine $\mathbf{\hat{x}}$ uniquely from $\mathbf{y} =  \mathbf{DV}\mathbf{\hat{x}}$, the matrices $\mathbf{A}$ should not have a non-zero null space vector. 
\begin{figure*}[!t]
	\begin{align}
	\mathbf{V}= \begin{pmatrix}
	e^{-j K \omega_0 t_1} & \cdots &e^{-j \omega_0 t_1}& t_1 & e^{j \omega_0 t_1} & \cdots& e^{j K \omega_0 t_1} \\
	e^{-j K \omega_0 t_2} & \cdots &e^{-j \omega_0 t_2}& t_2 & e^{j \omega_0 t_2} &  \cdots & e^{j K \omega_0 t_2} \\
	\vdots & \ddots & \vdots & \vdots & \vdots & \ddots & \vdots \\
	e^{-j K \omega_0 t_N} & \cdots &e^{-j \omega_0 t_N} & t_N & e^{j \omega_0 t_N} & \cdots & e^{j K \omega_0 t_N} \\
	\end{pmatrix} \in \mathbb{C}^{N \times 2K+1}.
	\label{eq:vmat}
	\end{align}
\end{figure*}
\begin{figure*}[!h]
	\begin{align}
	\underbrace{
		\begin{pmatrix}
		e^{-j K \omega_0 t_1} & \cdots &e^{-j \omega_0 t_1}& 1 & e^{j \omega_0 t_1} & \cdots& e^{j K \omega_0 t_1} \\
		e^{-j K \omega_0 t_2} & \cdots &e^{-j \omega_0 t_2}& 1 & e^{j \omega_0 t_2} &  \cdots & e^{j K \omega_0 t_2} \\
		\vdots & \ddots & \vdots & \vdots & \vdots & \ddots & \vdots \\
		e^{-j K \omega_0 t_N} & \cdots &e^{-j \omega_0 t_N} & 1 & e^{j \omega_0 t_N} & \cdots & e^{j K \omega_0 t_N} \\
		\end{pmatrix}}_{\mathbf{W}}
	\underbrace{\begin{pmatrix}
		\bar{x}[-K] \\
		\vdots \\
		\bar{x}[-1]\\
		-c \\
		\bar{x}[1]\\
		\vdots \\
		\bar{x}[K]
		\end{pmatrix}}_{\mathbf{\bar{x}}} = 
	-\bar{x}[0]
	\underbrace{\begin{pmatrix}
		t_1 \\t_2\\ \vdots \\t_N
		\end{pmatrix}}_{\mathbf{t}}.
	\label{eq:Wx=t}
	\end{align}
\end{figure*}

The matrix $\mathbf{D}$ is a difference operator which has null space vector $c \mathbf{1}_N$ where $c \in \mathbb{C}\backslash \{0\}$. Hence, if there exist a non-zero vector $\mathbf{x}$  as in \eqref{eq:xhat} whose components satisfy \eqref{eq:complex_conj}, such that $\mathbf{Vx} = c \mathbf{1}_N$ for some arbitrary $c$, then there does not exist a unique solution.  We show that for $N \geq 2K+2 \geq 2L+1$, uniqueness is guaranteed. Specifically, we would like to show that there does not exist an $\mathbf{x}$ satisfying \eqref{eq:complex_conj}, a set $\{t_n\}_{n=1}^N$, and $c\neq0$ such that $\mathbf{Vx} = c \mathbf{1}_N$ for $N > 2K+1$.

For simplicity of discussion, we define $\bar{x}[k] = \frac{\hat{x}[k]}{j k \omega_0}$ for $k \neq 0$ and $\bar{x}[0] = \hat{x}[0]$. The modulus and angle of the complex-valued coefficient $x[k]$ is denoted by $|x[k]|$ and $\angle x[k]$, respectively.
The equation $\mathbf{Vx} = c \mathbf{1}_N$ is re-written as in \eqref{eq:Wx=t} or alternatively as

\begin{align}
f(t) = r(t), \quad \text{for} \quad t = t_1, \cdots, t_N,
\label{eq:f=r}
\end{align}
where
\begin{align}
f(t) = c - \sum_{k = 1}^K |\bar{x}[k]|\, \cos(k\omega_0 t + \angle \bar{x}[k])  
\label{eq:ft}
\end{align}
is a $K$-th order trigonometric polynomial and $r(t) = \bar{x}[0] t$ is a straight line with slope $ \bar{x}[0]$.
If there exist a $c$ and $\{\bar{x}[k]\}_{k = 0}^K$ such that $f(t)$ and $r(t)$ intersect each other $N$-times within an interval $(0, T]$, that is, there exists a set $\mathcal{T}_N = \{t_n \in [0, T), n =1, \cdots, N \}$ satisfying $f(t)=r(t)$ then uniqueness is not guaranteed.

\begin{figure}
	\centering
	\includegraphics[width= 3 in]{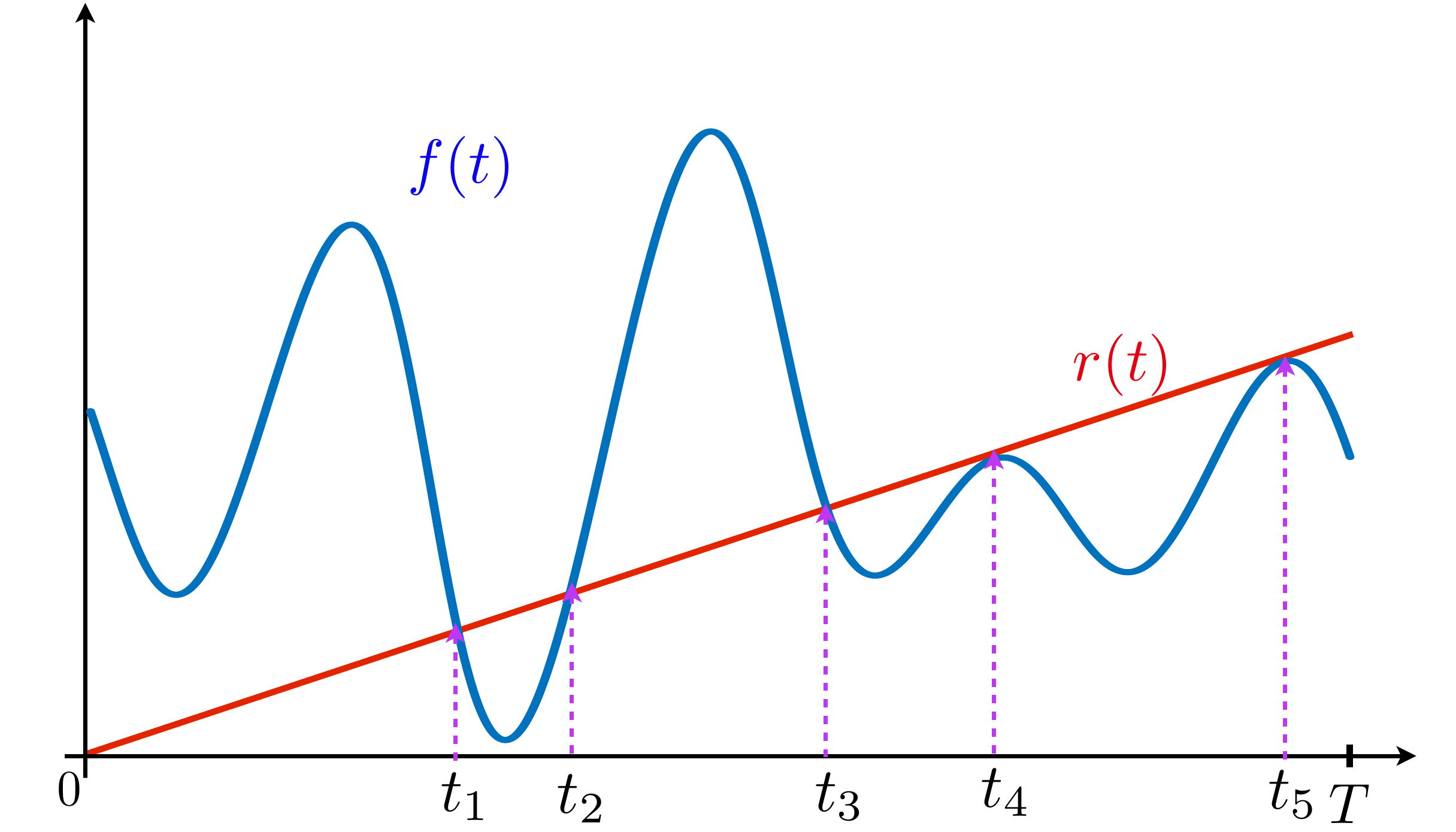}
	\caption{An illustration of intersection of a trigonometric polynomial $f(t)$ and a straight line $r(t)$.}
	\label{fig:intersection}
\end{figure}

Let us consider two mutually exclusive cases: (1) $\bar{x}[0] = \hat{x}[0] = 0$ and (2) $\bar{x}[0] = \hat{x}[0] \neq 0$.

\noindent \emph{Case-1:} 
For $\bar{x}[0] = 0$, the slope of the straight line $r(t)$ is zero and hence, \eqref{eq:f=r} is equivalent to determining zeros of $f(t)$ in the interval $(0, T])$. Since $f(t)$ is a trigonometric polynomial of order $K$ with $\omega_0 = \frac{2\pi}{T}$, it will have a maximum of $2K$ zeros within the interval $(0, T])$  \cite[p. 150]{powell_1981}. Hence, for $N>2K$, there does not exist  a $c\neq 0$ and a feasible $\{\bar{x}[k]\}_{k = 0}^K$  such that \eqref{eq:f=r} holds true. 

\noindent \emph{Case-2:} 
Consider the case when $\bar{x}[0] \neq 0$. We intend to determine the maximum number of intersections of a trigonometric polynomial of order $K$ with a straight line. To this end, let $t_1$ be the first intersection point. Further, let us assume that the slope of $f(t)$ at $t=t_1$ is positive (or negative), that is, $f^{\prime}(t_1) >0$ (or $f^{\prime}(t_1) <0$) where $f^{\prime}(t)$ denotes derivative of $f(t)$. This implies that there may exist a maximum (or minimum) of $f(t)$ for $t \in [0,  t_1)$ and a minimum (or maximum) for $t \in (t_1, T )$. An illustrative example is shown in Fig.~\ref{fig:intersection}. Since $r(t)$ is a monotone function, to have a second intersection $t_2$, it is necessary that $f^{\prime}(t)$ changes its sign. In essence, if there exist a $t_2 \in \mathcal{T}_N$, then there should be a minimum (or maximum) of $f(t)$ in the interval $(t_1, t_2)$. Applying this argument to the remaining intersection points in $\mathcal{T}_{N}$ we infer that for $N$ intersection points there should be atleast $N-1$ extrema. Alternatively, a function with $N-1$ extrema can intersect a monotone function at a maximum of $N$ points. As $f^{\prime}(t)$ too is a $K$-th order trigonometric polynomial, it has a maximum of $2K$ zeros \cite[p. 150]{powell_1981}. This implies that $f(t)$ can have a maximum of $2K$ extrema. Hence, $f(t)$ can intersect $r(t)$ at a maximum of $2K+1$ points within interval $(0, T])$. This implies that for $N>2K+1$, the equation $f(t) = r(t)$ can not have any solution and hence, matrix $\mathbf{A}$ does not have a non-zero null-vector and uniqueness is guaranteed.

\end{proof}


\bibliographystyle{ieeetr}
\bibliography{refs}

\begin{thebibliography}{10}

\bibitem{eldar2015sampling}
Y.~C. Eldar, {\em Sampling theory: Beyond bandlimited systems}.
\newblock Cambridge University Press, 2015.

\bibitem{unser2000sampling}
M.~Unser, ``Sampling-50 years after shannon,'' {\em Proc. IEEE}, vol.~88,
  no.~4, pp.~569--587, 2000.

\bibitem{nyquist1928certain}
H.~Nyquist, ``Certain topics in telegraph transmission theory,'' {\em Trans.
  American Inst. of Elect. Eng.}, vol.~47, no.~2, pp.~617--644, 1928.

\bibitem{eldar2009compressed}
Y.~C. Eldar, ``Compressed sensing of analog signals in shift-invariant
  spaces,'' {\em IEEE Trans. Signal Process.}, vol.~57, no.~8, pp.~2986--2997,
  2009.

\bibitem{christensen2004oblique}
O.~Christensen and Y.~C. Eldar, ``Oblique dual frames and shift-invariant
  spaces,'' {\em Applied and Comput. Harmonic Anal.}, vol.~17, no.~1,
  pp.~48--68, 2004.

\bibitem{asynchronous_adc}
M.~Malmirchegini, M.~M. Kafashan, M.~Ghassemian, and F.~Marvasti, ``Non-uniform
  sampling based on an adaptive level-crossing scheme,'' {\em IET Signal
  Process.}, vol.~9, pp.~484--490(6), August 2015.

\bibitem{lazar2004perfect}
A.~A. Lazar and L.~T. T{\'o}th, ``Perfect recovery and sensitivity analysis of
  time encoded bandlimited signals,'' {\em IEEE Trans. Circuits Syst. I: Reg.
  Papers}, vol.~51, no.~10, pp.~2060--2073, 2004.

\bibitem{adam2020sampling}
K.~Adam, A.~Scholefield, and M.~Vetterli, ``Sampling and reconstruction of
  bandlimited signals with multi-channel time encoding,'' {\em IEEE Tran.
  Signal Process.}, vol.~68, pp.~1105--1119, 2020.

\bibitem{rastogi2009low}
M.~Rastogi, V.~Garg, and J.~G. Harris, ``Low power integrate and fire circuit
  for data conversion,'' in {\em proc. IEEE Int. Symp. Circuits and Syst.},
  pp.~2669--2672, IEEE, 2009.

\bibitem{alexandru2019reconstructing}
R.~Alexandru and P.~L. Dragotti, ``Reconstructing classes of non-bandlimited
  signals from time encoded information,'' {\em IEEE Trans. Signal Process.},
  vol.~68, pp.~747--763, 2019.

\bibitem{hilton2021time}
M.~Hilton, R.~Alexandru, and P.~L. Dragotti, ``Time encoding using the
  hyperbolic secant kernel,'' in {\em European Signal Process. Conf.
  (EUSIPCO)}, pp.~2304--2308, IEEE, 2021.

\bibitem{rudresh2020time}
S.~Rudresh, A.~J. Kamath, and C.~S. Seelamantula, ``A time-based sampling
  framework for finite-rate-of-innovation signals,'' in {\em Proc. IEEE Int.
  Conf. Acoust., Speech and Signal Process. (ICASSP)}, pp.~5585--5589, 2020.

\bibitem{adam2020encoding}
K.~Adam, A.~Scholefield, and M.~Vetterli, ``Encoding and decoding mixed
  bandlimited signals using spiking integrate-and-fire neurons,'' in {\em Proc.
  IEEE Int. Conf. Acoust., Speech and Signal Process. (ICASSP)},
  pp.~9264--9268, IEEE, 2020.

\bibitem{martinez2019delta}
P.~Mart{\'\i}nez-Nuevo, H.-Y. Lai, and A.~V. Oppenheim, ``Delta-ramp encoder
  for amplitude sampling and its interpretation as time encoding,'' {\em IEEE
  Trans. Signal Process.}, vol.~67, no.~10, pp.~2516--2527, 2019.

\bibitem{neuromorphic_computer}
B.~{Rajendran}, A.~{Sebastian}, M.~{Schmuker}, N.~{Srinivasa}, and
  E.~{Eleftheriou}, ``Low-power neuromorphic hardware for signal processing
  applications: A review of architectural and system-level design approaches,''
  {\em IEEE Signal Process. Mag.}, vol.~36, no.~6, pp.~97--110, 2019.

\bibitem{camera1}
F.~{Barranco}, C.~{Fermüller}, and Y.~{Aloimonos}, ``Contour motion estimation
  for asynchronous event-driven cameras,'' {\em Proc. IEEE}, vol.~102, no.~10,
  pp.~1537--1556, 2014.

\bibitem{camera2}
F.~{Barranco}, C.~{Fermüller}, and Y.~{Aloimonos}, ``Contour motion estimation
  for asynchronous event-driven cameras,'' {\em Proc. IEEE}, vol.~102, no.~10,
  pp.~1537--1556, 2014.

\bibitem{lazar2004time}
A.~A. Lazar, ``Time encoding with an integrate-and-fire neuron with a
  refractory period,'' {\em Neurocomputing}, vol.~58, pp.~53--58, 2004.

\bibitem{lazar2003time}
A.~A. Lazar and L.~T. T{\'o}th, ``Time encoding and perfect recovery of
  bandlimited signals,'' in {\em Proc. IEEE Int. Conf. Acoust., Speech and
  Signal Process. (ICASSP)}, vol.~6.

\bibitem{feichtinger2012approximate}
H.~G. Feichtinger, J.~C. Pr{\'\i}ncipe, J.~L. Romero, A.~S. Alvarado, and G.~A.
  Velasco, ``Approximate reconstruction of bandlimited functions for the
  integrate and fire sampler,'' {\em Advances in Computational Math.}, vol.~36,
  no.~1, pp.~67--78, 2012.

\bibitem{adam2019multi}
K.~Adam, A.~Scholefield, and M.~Vetterli, ``Multi-channel time encoding for
  improved reconstruction of bandlimited signals,'' in {\em Proc. IEEE Int.
  Conf. Acoust., Speech and Signal Process. (ICASSP)}, pp.~7963--7967, IEEE,
  2019.

\bibitem{vetterli2002sampling}
M.~Vetterli, P.~Marziliano, and T.~Blu, ``Sampling signals with finite rate of
  innovation,'' {\em IEEE Trans. Signal Process.}, vol.~50, no.~6,
  pp.~1417--1428, 2002.

\bibitem{tur2011innovation}
R.~Tur, Y.~C. Eldar, and Z.~Friedman, ``Innovation rate sampling of pulse
  streams with application to ultrasound imaging,'' {\em IEEE Trans. Signal
  Process.}, vol.~59, no.~4, pp.~1827--1842, 2011.

\bibitem{urigiien2012sampling}
J.~A. Urigiien, Y.~C. Eldar, P.~Dragotti, and Z.~Ben-Haim, ``Sampling at the
  rate of innovation: Theory and applications,'' {\em Compressed Sensing:
  Theory and Applications}, vol.~148, 2012.

\bibitem{bar2014sub}
O.~Bar-Ilan and Y.~C. Eldar, ``Sub-{N}yquist radar via {D}oppler focusing,''
  {\em IEEE Trans. Signal Process.}, vol.~62, no.~7, pp.~1796--1811, 2014.

\bibitem{bajwa2011identification}
W.~U. Bajwa, K.~Gedalyahu, and Y.~C. Eldar, ``Identification of parametric
  underspread linear systems and super-resolution radar,'' {\em IEEE Trans.
  Signal Process.}, vol.~59, no.~6, pp.~2548--2561, 2011.

\bibitem{rudresh2017finite}
S.~Rudresh and C.~S. Seelamantula, ``Finite-rate-of-innovation-sampling-based
  super-resolution radar imaging,'' {\em IEEE Trans. Signal Process.}, vol.~65,
  no.~19, pp.~5021--5033, 2017.

\bibitem{mulleti2014ultrasound}
S.~Mulleti, S.~Nagesh, R.~Langoju, A.~Patil, and C.~S. Seelamantula,
  ``Ultrasound image reconstruction using the finite-rate-of-innovation
  principle,'' in {\em IEEE Int. Conf. Image Process. (ICIP)}, pp.~1728--1732,
  IEEE, 2014.

\bibitem{mulleti2016fri}
S.~Mulleti, B.~A. Shenoy, and C.~S. Seelamantula, ``{FRI} sampling on
  structured nonuniform grids—application to super-resolved optical
  imaging,'' {\em IEEE Trans. Signal Process.}, vol.~64, no.~15,
  pp.~3841--3853, 2016.

\bibitem{donoho2006compressed}
D.~L. Donoho, ``Compressed sensing,'' {\em IEEE Trans. Info Theory}, vol.~52,
  no.~4, pp.~1289--1306, 2006.

\bibitem{stoica}
P.~Stoica and R.~L. Moses, {\em Introduction to Spectral Analysis}.
\newblock Upper Saddle River, NJ: Prentice Hall, 1997.

\bibitem{gontier2014sampling}
D.~Gontier and M.~Vetterli, ``Sampling based on timing: Time encoding machines
  on shift-invariant subspaces,'' {\em Applied and Comput. Harmonic Anal.},
  vol.~36, no.~1, pp.~63--78, 2014.

\bibitem{powell_1981}
M.~J.~D. Powell, {\em Approximation Theory and Methods}.
\newblock Cambridge University Press, 1981.

\end{thebibliography}

\end{document}